\documentclass{sig-alternate}

\usepackage{amsthm,amssymb,amsmath,stmaryrd}
\usepackage{graphics}
\usepackage{bbm}
\usepackage{tikz}
\usepackage[plainpages=false,pdfpagelabels=false,colorlinks=true,citecolor=blue,hypertexnames=false]{hyperref}
\usepackage{color}
\usepackage{graphicx}
\usepackage{amsmath}
\usepackage{amssymb}
\usepackage{enumerate}
\usepackage{mathrsfs}
\usepackage{makecell}
\usepackage[noend]{algpseudocode}
\usepackage{multirow}
\usepackage{mathptmx}

\usepackage{dsfont}
\usepackage{epsfig}
\usepackage{mathrsfs}
\usepackage{amsmath}
\usepackage{cases}

\usepackage[vlined,ruled,linesnumbered,boxed]{algorithm2e}

\newtheorem{theorem}{Theorem}[section]
\newtheorem{example}[theorem]{Example}
\newtheorem{corollary}[theorem]{Corollary}
\newtheorem{lemma}[theorem]{Lemma}
\newtheorem{proposition}[theorem]{Proposition}
\newtheorem{definition}[theorem]{Definition}

\makeatletter \@addtoreset{equation}{section}

\newcommand{\x}{\mathbf{x}}

\begin{document}

\clubpenalty=10000 \widowpenalty =10000

\title{An Effective Framework for Constructing Exponent Lattice Basis of Nonzero Algebraic Numbers}

\numberofauthors{1}

\author{\medskip
Tao Zheng and Bican Xia\\
       \smallskip
       \affaddr{School of Mathematical Sciences, Peking University}\\
       \smallskip
      \email{1601110051@pku.edu.cn, xbc@math.pku.edu.cn}
}

\maketitle
\begin{abstract}
Computing a basis for the exponent lattice of algebraic numbers is a basic problem in the field of computational number theory with applications to many other areas. The main cost of a well-known algorithm \cite{ge1993algorithms,kauers2005algorithms} solving the problem is on computing the primitive element of the extended field generated by the given algebraic numbers. When the extended field is of large degree, the problem seems intractable by the tool implementing the algorithm. In this paper, a special kind of exponent lattice basis is introduced. An important feature of the basis is that it can be inductively constructed, which allows us to deal with the given algebraic numbers one by one when computing the basis. Based on this, an effective framework for constructing exponent lattice basis is proposed. Through computing a so-called pre-basis first and then solving some linear Diophantine equations, the basis can be efficiently constructed. A new certificate for multiplicative independence and some techniques for decreasing degrees of algebraic numbers are provided to speed up the computation. The new algorithm has been implemented with Mathematica and its effectiveness is verified by testing various examples. Moreover, the algorithm is applied to program verification for finding invariants of linear loops.
\end{abstract}





\section{Introduction}

For $x=(x_1,\ldots,x_n)^T$$\in(\overline{\mathbb{Q}}^*)^n$, where $\overline{\mathbb{Q}}$ denotes the set of algebraic numbers, vectors $\alpha=(k_1,\ldots,k_n)^T$$\in\mathbb{Z}^n$ satisfying $x^\alpha$$= x_1^{k_1}x_2^{k_2}\cdots x_n^{k_n}=1$ form an \emph{exponent lattice}, which accepts a basis since $\mathbb{Z}$ is Noetherian. Computing that basis is a significant problem from both practical  and theoretical points of view.

Exponent lattice has various applications. For example, based on computing exponent lattice, an algorithm was proposed to compute the Zariski closure of a finitely generated group of invertible matrices in \cite{derksen2005quantum}, the growth behavior when $k$$\rightarrow $$\infty$ of rational linear recurrence sequence was studied in \cite{recurrence}, the problem of finding integers $k_1,\ldots,k_n$ \emph{s.t.} $\mathcal{G}_1^{k_1}\cdots\mathcal{G}_n^{k_n}$ is a rational function for $G-$solutions $\mathcal{G}_j$ in  \cite{chen2011structure} could be solved, also an algorithm was provided to compute the ideal of algebraic relations among C-finite sequences in \cite{kauers2008computing}. A class of loop invariants called L-invariants introduced in \cite{lvov2010polynomial} are closely related to exponent lattice: each vector in the lattice corresponds to an L-invariant. 
Moreover Theorems 3 and 5 in \cite{structure} show that a part of the invariant ideal of a linear loop is exactly the lattice ideal defined by the exponent lattice. The lattice ideal accepts a set of finite generators corresponding to a Markov basis of the lattice which can be computed from the usual basis by the method in \cite{markov}.

The first result leading to the computability of the exponent lattice basis was presented in \cite{masser1988linear}, 
which bounds a basis of the lattice inside a box that contains finitely many vectors, allowing one to do exhaustive search inside the box to obtain a basis of the lattice. A more efficient algorithm to solve this problem was proposed in \cite{ge1993algorithms}. It is redescribed in \cite{kauers2005algorithms} $\S7.3$ and implemented as a Mathematica package \textbf{FindRelations}.

According to \cite{kauers2005algorithms}, ``the runtime is usually negligible compared to the time needed for computing the primitive element''. That means \textbf{FindRelations} usually takes much time to compute a primitive element of the extended field generated by the given algebraic numbers. When degree of the extended field is slightly large, the function tends to fail to return an answer within one hour. 
In this paper, a special kind of exponent lattice basis is introduced. An important feature of the basis is that it can be inductively constructed, which allows us to deal with the given algebraic numbers one by one when computing the basis. Based on this, an effective framework for constructing exponent lattice basis is proposed. Through computing a so-called pre-basis first and then solving some linear Diophantine equations, the basis can be efficiently constructed. A new certificate for multiplicative independence and some techniques for decreasing degrees of algebraic numbers are provided to speed up the computation.

The paper is organized as follows. Section \ref{sec:def} introduces basic conceptions used throughout the paper. The main result on constructing a basis of the lattice (Theorem \ref{basisthm}) and a certificate for multiplicative independence (Theorem \ref{suff}) are given in Section \ref{sec:main}. Main algorithms computing the pre-basis and the basis are then presented in Section \ref{sec:alg}. Section \ref{sec:red} devotes to design algorithms decreasing degrees  of algebraic numbers. Experiments are carried out and an application to compute invariant ideal of linear loops is illustrated in Section \ref{sec:example}. Finally, the paper is concluded in Section \ref{sec:conclusions}.

\section{Definitions and conceptions}\label{sec:def}
\begin{definition}
A sequence $\{x_i\}_{i=1}^n\subset{\overline{\mathbb{Q}}}^*$ is \emph{multiplicatively dependent} \emph{(dependent} for short\emph{)} if  there are integers $k_1,\ldots,k_n$ not all zero such that $ x_1^{k_1}\cdots x_n^{k_n}=1$, otherwise it is \emph{multiplicatively independent} \emph{(independent} for short\emph{)}.
\end{definition}
\noindent{}\emph{Note:} An empty sequence $\epsilon$ containing no algebraic numbers is multiplicatively independent with \textbf{length}($\epsilon)=0$, and is a subsequence of any sequence of nonzero algebraic numbers.
\begin{definition}
\label{mis}
For sequence $\{x_i\}_{i=1}^n\subset{\overline{\mathbb{Q}}}^*$, its subsequence $S$ is a \emph{maximal independent sequence} if : (i) $S$ is multiplicatively independent and (ii) \textbf{\emph{length}}$(S)<n$ implies any supersequence $T$ of $S$ with \textbf{\emph{length}}$(T)>$\textbf{\emph{length}}$(S)$ is multiplicatively dependent.
\end{definition}

\noindent{}\emph{Note: }$\epsilon$ is a maximal independent sequence of $x_1,\ldots,x_n$ iff every $x_j$ is a root of unity.
\begin{definition}
 A number $\mathfrak{a}\in\overline{\mathbb{Q}}^*$ can be \emph{pseudo}-\emph{multiplicatively represented} by $x_1,\ldots,x_n\in\overline{\mathbb{Q}}^*$ if  there are integers $k_1,\ldots,k_n$ and $k\neq0$ such that $\mathfrak{a}^k=x_1^{k_1}\cdots x_n^{k_n}$.  If $n=0$, $\mathfrak{a}$ can  be \emph{pseudo}-\emph{multiplica}-\emph{tively} \emph{represented} by $\epsilon$ means $\mathfrak{a}$ is a root of unity.
\end{definition}
\begin{definition}
The \emph{order} of a root of unity $\mathfrak{a}$, denoted by $\textbf{\emph{Order}}(\mathfrak{a})$, is the least positive integer so that $\mathfrak{a}^k=1$.
\end{definition}
\begin{proposition}
\label{pmr}
Suppose $S$ is a maximal independent sequence of $\{x_i\}_{i=1}^n\subset{\overline{\mathbb{Q}}}^*$. Then $x_i$ can be pseudo-multiplicatively represented by $S$ for each $i=1,...,n$.
\end{proposition}
\begin{lemma}
\label{rank}
Let $S=\{y_i\}_{i=1}^m\subset\overline{\mathbb{Q}}^*$ and $T= \{x_j\}_{j=1}^n\subset\overline{\mathbb{Q}}^*$. If every $y_j$ can be pseudo-multiplicatively represented by $T$ and $m>n$, then $S$ is multiplicatively dependent.
\end{lemma}
\begin{proposition}
Let $S$ and $T$ be two maximal independent sequences of $x_1,\ldots,x_n$, then $\emph{\textbf{length}}(S)=\emph{\textbf{length}}(T)$.
\end{proposition}
\begin{definition}
\label{algerank}
The length of any maximal independent sequence of $\{x_i\}_{i=1}^n\subset\overline{\mathbb{Q}}^*$ is called the \emph{rank} of $x_1,\ldots,x_n$, denoted by $\textbf{\emph{rank}}(x)$.
\end{definition}
\begin{definition}
\label{rx}
Define $\mathcal{R}_x=\{v\in\mathbb{Z}^n|x^v=1\}$ with $x=(x_1,\ldots,x_n)^T$ in $({\overline{\mathbb{Q}}^*})^n$ and $x^v = x_1^{k_1}\cdots x_n^{k_n}$ if $v = (k_1,\ldots,k_n)^T$. $\mathcal{R}_x$ is called \emph{the exponent lattice} of $x$. The elements of $\mathcal{R}_x\backslash\{\mathbf{0}\}$ are called \emph{dependent vectors} of $x$. A $\mathbb{Z}$-independent finite set $\{\alpha_1,\ldots,\alpha_r\}\subset\mathcal{R}_x$ is called \emph{a basis} of $\mathcal{R}_x$ if $\forall v\in\mathcal{R}_x, \exists k_1,\ldots,k_r\in\mathbb{Z}$ s.t. $v=\sum_{i=1}^r k_i\alpha_i$, where $\mathbb{Z}$-independent means $\ell_1\alpha_1+\cdots+\ell_r\alpha_r=\mathbf{0}$ implies $\ell_1=\cdots=\ell_r=0$ for integers $\ell_1,\ldots,\ell_r$.
\end{definition}

\section{Main results}\label{sec:main}
\subsection{Main Theorems for Constructing a Basis}
For a vector $v=(z_1,\ldots,z_n)^T\in \mathbb{C}^n$, if $1\leq k\leq n$, denote $v|k= (z_1,\ldots,z_k)^T$ and $v(k)$ the coordinate of $v$ indexed by $k$. Denote $[j]=\{1,\ldots,j\}$ for integer $j\geq 1$. Set $I\subset[n]$, denote by $\mathbb{Z}^I$ the lattice consists of all integer vectors whose coordinates are indexed by $I$. In addition, we simplify the notation $\mathbb{Z}^{[j]}$ to $\mathbb{Z}^j$. For any two subsets $I\subset J\subset [n]$, each vector in $\mathbb{Z}^I$ is considered to be a  vector in $\mathbb{Z}^J$ with extra coordinates, indexed by $J\backslash I$, equal to $0$. This ambiguity does not lead to any troubles but brings us convenience. For $x=(x_1,\ldots,x_n)^T$$\in(\overline{\mathbb{Q}}^*)^n$ and $v\in\mathbb{Z}^{I}$, define $x^v=\prod_{i\in I}x_i^{v(i)}$. This value remains the same if one considers $v$ to be a vector in $\mathbb{Z}^J$,  for  any $J\supset I$, and computes as $x^v=\prod_{j\in J}x_j^{v(j)}$.
\begin{definition}
\label{bv}
For $x\in(\overline{\mathbb{Q}}^*)^n$ and each  $j\in[n]$, define
 {(i)} $V_j=\{v\in\mathbb{Z}^j| v(j)>0,\; x^v=1\}$,
{(ii)} $J=\{j\in [n]|V_j\neq\emptyset\}$,
{(iii)}
$\textbf{\emph{min}}_j=\min_{v\in V_j} v(j)$
 and {(iv)}
$\tilde{V}_j=\{v\in V_j|v(j)=\textbf{\emph{min}}_j\}$ for each $j\in J$. 
Then $\tilde{V}_j\neq\emptyset$ for any $j\in J$. A set of vectors of the form $
\{u_j\}_{j\in J}$, where $u_j\in\tilde{V}_j$ for each $j\in J$,
is called a set of \emph{basis vectors}.
\end{definition}
The name is justified by Theorem \ref{basisthm} that is to come.
\begin{proposition}
\label{divide}
$\forall j\in J$, $\forall v \in V_j$, \textbf{\emph{min}}$_j=u_j(j)$ divides $v(j)$.
\end{proposition}
\begin{proof}
Write $v_j(j)= q\cdot \textbf{min}_j+r$, where $0\leq r<\textbf{min}_j$. Assume that $0<r<\textbf{min}_j$, then $(v_j-qu_j)(j)=r>0$ and $(v_j-qu_j)\in V_j$. Then $r<\textbf{min}_j$ contradicts Definition \ref{bv} (iii).
\end{proof}
In the following, we fix a certain set of basis vectors $\{u_j\}_{j\in J}$ defined in Definition \ref{bv}. Define inductively $n+1$ sets of integer vectors as follows:
\begin{equation}
\begin{array}{l}
{B_0} = \emptyset ,\\
{B_j} = \left\{ {\begin{array}{*{20}{c}}
{{B_{j - 1}},\;\;\;\;\;\;\;\;\;\;\;\;\;\;\;\;\;{\rm{if}}\;{V_j}  =  \emptyset},\\
{{B_{j - 1}} \cup \{ {u_j}\} ,\;\;\;\;{\rm{if}}\;{V_j} \ne \emptyset },
\end{array}} \right.
\end{array}
\label{Basis}
\end{equation}
\noindent{}for each $j\in[n]$. We have the following main theorem:
\begin{theorem}
\label{basisthm}
Consider $B_n$ to be a subset of$\;\mathbb{Z}^n$, then $B_n=\{u_j\}_{j\in J}$ is a basis of the exponent lattice $\mathcal{R}_x$.
\end{theorem}
\begin{proof}
If $B_n=\emptyset$, then $\mathbb{R}_x=\{\mathbf{0}\}$ and the conclusion holds. Define $\textbf{Tail}(v)=\max\{i\in[n]|v(i)\neq0\}$. We use induction on $\textbf{Tail}(v)$ to prove that $\forall v\in\mathcal{R}_x\backslash\{\mathbf{0}\}, v\in\textbf{span}(B_n)$. If $\textbf{Tail}(v)=1$, $x_1$ is a root of unity and $\textbf{Order}(x_1)|v(1)$. Thus $v(1)\in\textbf{span}\{u_1\}\subset B_n$. Suppose $\textbf{Tail}(v)=k\geq2$ and $\forall v\;(v\in\mathcal{R}_x\wedge\textbf{Tail}(v)<k)\Rightarrow v\in\textbf{span}(B_n)$. Then by Proposition \ref{divide}, $u_k(k)$ divides $|v(k)|$. Let $\tilde{v}= v-\frac{v(k)}{u_k(k)}\cdot u_k\in\mathcal{R}_x$, then $\textbf{Tail}(\tilde{v})<k$. Therefore $\tilde{v}\in\textbf{span}(B_n)$ and $v=\tilde{v}+v(k)/u_k(k)\cdot u_k\in\textbf{span}(B_n)$. Moreover, the vectors in $B_n$ are obviously $\mathbb{Z}$-independent.
\end{proof}

 Theorem \ref{basisthm} indicates how one  constructs inductively a basis of $\mathcal{R}_x$. The point is to compute $u_j$. Before that, we digress a little to see the main idea to obtain a maximal independent sequence of the numbers $x_1,x_2,\ldots,x_n$.

Set:
\begin{equation}\begin{array}{l}
{S_0} =  \epsilon ,\\
{S_j} =  \left\{ {\begin{array}{*{20}{l}}
{{S_{j - 1}},\;\;\;\;\;\;\;\;{\rm{if}}\;{S_{j - 1}},{x_j}\;{\rm{is\; dependent,}}}\;\\
{{S_{j - 1}},{x_j},\;\;{\rm{if}}\;{S_{j - 1}},{x_j}\;{\rm{is \;independent,}}}
\end{array}} \right.
\end{array}\label{Sj}\end{equation}
\noindent{}\noindent{}for each $j\in[n]$. Here $S_{j - 1},{x_j}$ is the sequence obtained by attatching $x_j$ to the tail of the sequence $S_{j-1}$. The sequence $\epsilon,x_j$, for instance, means $x_j$.  Also we define for $j\in[n]$:
\begin{equation}\begin{array}{l}
{I_0}  =  \emptyset,\\
{I_j}  =  \left\{ {\begin{array}{*{20}{l}}
{{I_{j - 1}},\;\;\;\;\;\;\;\;\;\;\;\;\;{\rm{if}}\;{S_{j - 1}},{x_j}\;{\rm{is\; dependent,}}}\\
{{I_{j - 1}}\cup\{j\},\;\;{\rm{if}}\;{S_{j - 1}},{x_j}\;{\rm{is \;independent.}}}
\end{array}} \right.
\end{array}
\label{Ij}\end{equation}
\begin{theorem}
\label{misthm}
$S_n$ is a maximal independent sequence of $x_1,\ldots,x_n$.
\end{theorem}

The following proposition connects Theorem \ref{basisthm} and \ref{misthm}.
\begin{proposition}
\label{connection}
$V_j\neq\emptyset\Leftrightarrow$ the sequence $S_{j-1},x_{j}$ is dependent.
\end{proposition}
\begin{proof}
``$\Rightarrow$\text{'':} Set $v=(-k_1,-k_2,\ldots,-k_{j-1},k_j)^T$$\in V_j$, \emph{s.t.} $k_j>0$ and $x^v=1$, thus $x_j^{k_j}=x_1^{k_1}x_2^{k_2}\cdots x_{j-1}^{k_{j-1}}$. Set
$
{S_{j - 1}}\text{to be }{x_{{i_1}}},{x_{{i_2}}}, \cdots ,{x_{{i_m}}},
$
as defined in (\ref{Sj}). By Theorem \ref{misthm} and Proposition \ref{pmr}, each of $x_1,\ldots,x_{j-1}$ can be pseudo-multiplicatively represented by $S_{j-1}$:
\[\left\{ {\begin{array}{l}
\;\;\;{x_1^{{r_1}} = x_{i_1}^{{a_{11}}} \cdots x_{i_m}^{{a_{m1}}}},\\
\;\;\;\;\;\;\;\;\;\; \vdots \\
{x_{j-1}^{{r_{j-1}}} = x_{i_1}^{{a_{1,{j-1}}}} \cdots x_{i_m}^{{a_{m,{j-1}}}}},
\end{array}} \right.\]
\noindent{}where $r_1,r_2,\ldots,r_{j-1}\in\mathbb{Z}^*$, $a_{st}\in\mathbb{Z}$, for $s\in[m],\;t\in[j-1]$. Define $\Pi=r_1 r_2\cdots r_{j-1}$ and $\pi_{t}=\Pi/r_{t}$, then
\[\begin{array}{l}
x_j^{\Pi  \cdot {k_j}} = x_1^{\Pi  \cdot {k_1}}x_2^{\Pi  \cdot {k_2}} \cdots x_{j - 1}^{\Pi  \cdot {k_{j - 1}}}\\

\;\;\;\;\;\;\;\;\;\;{\rm{ =  }}\;x_1^{{r_1} \cdot {k_1} \cdot {\pi _1}}x_2^{{r_2} \cdot {k_2} \cdot {\pi _2}} \cdots x_{j - 1}^{{r_{j - 1}} \cdot {k_{j - 1}} \cdot {\pi _{j - 1}}}\\

\;\;\;\;\;\;\;\;\;\; = {(x_{{i_1}}^{{a_{11}}} \cdots x_{{i_m}}^{{a_{m1}}})^{{k_1} \cdot {\pi _1}}} \cdots {(x_{{i_1}}^{{a_{1,j - 1}}} \cdots x_{{i_m}}^{{a_{m,j - 1}}})^{{k_{j - 1}} \cdot {\pi _{j - 1}}}}
\end{array}\]
\noindent{}with $\Pi\cdot k_j\neq0$, thus the sequence $S_{j-1},x_j$ is dependent.

``$\Leftarrow$'': Since $S_{j-1}$ is multiplicatively independent, $S_{j-1}$ is a maximal independent sequence of $S_{j-1},x_j$. According to Proposition \ref{pmr}, $x_j$ can be pseudo-multiplicative represented by $S_{j-1}$: $x_j^{l} = x_{i_1}^{l_1} \cdots x_{i_m}^{{l_m}}$
\noindent{}with $0\neq l, l_1,\ldots,l_m\in\mathbb{Z}$. Define $v\in \mathbb{Z}^{I_{j-1}\cup\{j\}}$ by $v(i_s)=l_s,\;$for $ s\in[m]$ and $v(j)=-l$. Then $v$ or $-v\in V_j$, hence $V_j\neq\emptyset$.
\end{proof}
Proposition \ref{connection} implies that \[(|B_j|-|B_{j-1}|)+(\textbf{length}(S_j)-\textbf{length}(S_{j-1}))=1\] for any $j$, thus $|B_n|-|B_0|+\textbf{length}(S_n)-\textbf{length}(S_{0})=n.$ Now $|B_n|=\textbf{rank}(\mathcal{R}_x)$ by Theorem \ref{basisthm}, $\textbf{length}(S_n)=\textbf{rank}(x)$ by Theorem \ref{misthm} and $|B_0|=\textbf{length}(S_{0})=0$, thus $\textbf{rank}(\mathcal{R}_x)+\textbf{rank}(x)=n$. We obtain the following theorem.
\begin{theorem}
For $x\in(\overline{\mathbb{Q}}^*)^n$, $\emph{\textbf{rank}}(\mathcal{R}_x)+\emph{\textbf{rank}}(x)=n$.
\end{theorem}
Note that  $I_{j-1}\cup\{j\}\subset [j]$, thus $\mathbb{Z}^{I_{j-1}\cup\{j\}}\subset\mathbb{Z}^j$ by convention. Define a subset of $V_j$ :
\begin{equation}
U_j=V_j\cap\mathbb{Z}^{I_{j-1}\cup\{j\}}
\end{equation}
for $j\in[n]$. It is clear, from the  ``$\Leftarrow$'' part of the proof of Proposition \ref{connection}, that $U_j\neq\emptyset$ whenever the sequence $S_{j-1},x_j$ is dependent. Indeed, the vector $v$ or $-v$ constructed there is in $U_j$. Conversely, if $U_j\neq\emptyset$, any vector in $U_j$ provides a dependent vector for the sequence $S_{j-1},x_j$. Hence  $U_j\neq\emptyset$ iff  the sequence $S_{j-1},x_j$ is dependent, which is equivalent to $V_j\neq\emptyset$. One sees $\{j\in[n]\;|\;U_j\neq\emptyset\}=J$ (Definition \ref{bv} (ii)).
\begin{definition}
\label{prebv}
A set of vectors of the form $
\{w_j\}_{j\in J}$, where $w_j\in U_j$ for each $j\in J$,
is called a set of \emph{pre-basis vectors}.
\end{definition}
In the following, we fix a certain pre-basis $\{w_j\}_{j\in J}$. Since $w_j\in U_j\subset V_j$, $u_j(j)$ divides $w_j(j)$ by Proposition \ref{divide}.

\subsection{A Certificate for Multiplicative Independence}

The degree of  an algebraic number $\mathfrak{a}$ is denoted by $\textbf{deg}(\mathfrak{a})$.
\begin{definition}
\label{nondegen}
 A sequence $\{x_i\}_{i=1}^n\subset\overline{\mathbb{Q}}^*$ is called \emph{degenerate} if $[\mathbb{Q}[x_1,x_2,\ldots,x_n]:\mathbb{Q}]<\prod_{i=1}^n\textbf{\emph{deg}}(x_i)$, otherwise \emph{non-degenerate}.
\end{definition}
\begin{theorem}
\label{suff}
If $\{x_i\}_{i=1}^n$$\subset\overline{\mathbb{Q}}^*$ is non-degenerate and $
\prod_{i\in[n]}x_i^{k_i}=1$ for integers $k_1,k_2,\cdots,k_n$, then $\forall i\in [n],\; x_i^{k_i}\in \mathbb{Q}.$
\end{theorem}
\begin{proof}
Suppose $\textbf{deg}(x_i)=d_i$ for $i\in[n]$. The non-degenerate condition implies that the set
$
\{x_1^{l_1}x_2^{l_2}\cdots x_n^{l_n}|\forall i \in [n],\;0\leq l_i\leq d_i-1\}
$
is a basis of the $\mathbb{Q}$-linear space $\mathbb{Q}[x_1,x_2,\ldots,x_n]$. For each $i\in[n]$, a basis of the $\mathbb{Q}$-linear space $\mathbb{Q}[x_i]$ is
$
\{1,x_i,x_i^2,\ldots,x_i^{d_i-1}\}.
$
Expanding $x_i^{k_i}$ along that basis we obtain:
\begin{equation}
x_i^{k_i}= a_{i0}+a_{i1}x_i+a_{i2}x_i^2+\cdots+a_{i,d_i-1}x_i^{d_i-1},
\label{span}
\end{equation}
where $a_{ij}\in \mathbb{Q}$, $i\in[n]$, $j+1\in[d_i]$. Then $\prod_{i\in[n]}x_i^{k_i}=1$ means:
\begin{equation}
\sum_{\forall i\in [n],\;0\leq l_i \leq d_i-1}a_{1l_1}a_{2l_2}\cdots a_{nl_n}x_1^{l_1}x_2^{l_2}\cdots x_n^{l_n}=1.
\label{compare}
\end{equation}
Since $\{x_1^{l_1}x_2^{l_2}\cdots x_n^{l_n}|\forall i \in [n], l_i+1\in[d_i]\}$ is a basis, two sides of (\ref{compare}) have exactly the same coefficients, {\it i.e.}
\begin{equation}\left\{ {\begin{array}{*{20}{c}}
a_{10}a_{20}\cdots a_{n0}=1,\;\;\;\;\;\;\;\;\;\;\;\;\;\;\;\;\;\;\;\;\;\;\;\;\;\;\;\;\;\;\;\;\\
a_{1l_1}a_{2l_2}\cdots a_{nl_n}=0,\;\text{if}\; \exists i\in[n],\;l_i>0,
\end{array}} \right.\label{coe}\end{equation}
Consider $a_{ij}\in\mathbb{Q}$ with $i\in[n]$, $j \in[d_i-1]$, then $a_{10}a_{20}\cdots a_{i-1,0}a_{ij}\cdot$\\$a_{i+1,0}\cdots a_{n0}=0$ by the $2^{\text{nd}}$ equality of (\ref{coe}). Thus $a_{ij}=0$ follows from the $1^{\text{st}}$ equality of (\ref{coe}). Then (\ref{span}) is reduced to $x_i^{k_i}=a_{i0}\in \mathbb{Q}$.
\end{proof}
\begin{definition}
\label{rorn}
Nonzero algebraic number $\mathfrak{a}$ is a \emph{root of rational} if $\exists k\in\mathbb{Z},k>0,\mathfrak{a}^k\in \mathbb{Q}$. The smallest such $k$ is \emph{the rational order} of $\mathfrak{a}$, denoted by $\textbf{\emph{Rorder}}(\mathfrak{a})$. For $\mathfrak{a}$ not a root of rational, define $\textbf{\emph{Rorder}}(\mathfrak{a})=0$. For convenience, define $\textbf{\emph{R}}(\mathfrak{a})= \mathfrak{a}^{\textbf{\emph{Rorder}}(\mathfrak{a})}$.
\end{definition}
\noindent{\emph{Note:}} If $\mathfrak{a}^\ell\in\mathbb{Q}$ for integer $\ell$, then $\textbf{Rorder}(\mathfrak{a})$ divides $\ell$.
\begin{corollary}
\label{suffcor}
If $x_1,\ldots,x_n$ satisfy non-degenerate condition in Definition \ref{nondegen} and none of them is a root of rational, then $x_1,\ldots,x_n$ are multiplicatively independent.
\end{corollary}
\begin{proof}
Set $x_1^{k_1}\cdots x_n^{k_n}=1$ for integers $k_j$. By Theorem \ref{suff},  $\forall i\in[n]$, $x_i^{k_i}\in \mathbb{Q}$. Since $x_i$ is not a root of rational, $k_i=0,\forall i\in[n]$.
\end{proof}
More generally, set \begin{equation} P=\{i\in[n]\;|\;x_i \text{ is a root of rational}\}=\{i_1,i_2,\ldots,i_s\}.\label{reducedequation}\end{equation} Suppose the non-degenerate condition holds, then $\prod_{i\in[n]}x_i^{k_i}=1$ implies $k_i=0$ for $i\notin P$. The problem is reduced to solving the equation $x_{i_1}^{k_{i_1}}x_{i_2}^{k_{i_2}}\cdots x_{i_s}^{k_{i_s}}=1$. In fact, $x_{i_1},x_{i_2},\ldots, x_{i_s}$ are  dependent iff $\textbf{R}(x_{i_1}),\textbf{R}(x_{i_2}),\ldots,\textbf{R}(x_{i_s})$ (Definition \ref{rorn}) are. More generally, for each $i\in [n]$, let $x_i^{{1}/{q_i}}$ be a certain $q_i$-th root of $x_i$. Formally define $x_i^{{p_i}/{q_i}}= (x_i^{{1}/{q_i}})^{p_i}$, for integers $q_i>0,p_i\neq 0$, then
\begin{proposition}
\label{scale}
$x_1^{{p_1}/{q_1}},x_2^{{p_2}/{q_2}},\ldots, x_n^{{p_n}/{q_n}}$ are multiplicatively dependent $\Leftrightarrow x_1,x_2,\ldots,x_n$ are multiplicatively dependent.
\end{proposition}
\begin{proof}
\emph{``$\Leftarrow$'':}  Set $(k_1,k_2,\ldots,k_n)^T\in \mathbb{Z}^n\backslash\{\mathbf{0}\}$ and $x_1^{k_1}x_2^{k_2}\cdots x_n^{k_n}=1$. Define $p=p_1p_2\cdots p_n$ and $\tilde{p}_i=p/p_i$, then $p=p_i\tilde{p}_i$ and
\[(x_1^{k_1}x_2^{k_2}\cdots x_n^{k_n})^p=x_1^{p_1\tilde{p}_1k_1}x_2^{p_2\tilde{p}_2k_2}\cdots x_n^{p_n\tilde{p}_nk_n}=1.\]
Thus $(x_1^{{p_1}/{q_1}})^{q_1\tilde{p}_1k_1}(x_2^{{p_2}/{q_2}})^{q_2\tilde{p}_2k_2}\cdots(x_n^{{p_n}/{q_n}})^{q_n\tilde{p}_nk_n}=1.$ Since $\exists i\in[n], k_{i}\neq0, q_{i}>0$ and $\tilde{p}_{i}\neq 0$, $q_{i}\tilde{p}_{i}k_i\neq0$. We find a dependent vector for $x_1^{{p_1}/{q_1}},x_2^{{p_2}/{q_2}},\ldots, x_n^{{p_n}/{q_n}}$.

\noindent\emph{``$\Rightarrow$'':} Set  $(k_1,k_2,\ldots,k_n)^T\in \mathbb{Z}^n\backslash{\{\mathbf{0}\}}$, such that  \begin{equation}(x_1^{p_1/q_1})^{k_1}(x_2^{p_2/q_2})^{k_2}\cdots (x_n^{p_n/q_n})^{k_n}=1.\label{relation}\end{equation} Define $q= q_1q_2\cdots q_n$ and $\tilde{q}_i=q/q_i$, then $q=q_i\tilde{q}_i$. Taking power to $q$ for both sides  of (\ref{relation}), we have
 \[(x_1^{1/q_1})^{q_1\tilde{q}_1p_1k_1}(x_2^{1/q_2})^{q_2\tilde{q}_2p_2k_2}\cdots (x_n^{1/q_n})^{q_n\tilde{q}_np_nk_n}=1.\]
Thus $x_1^{\tilde{q}_1p_1k_1}x_2^{\tilde{q}_2p_2k_2}\cdots x_n^{\tilde{q}_np_nk_n}=1$. Suppose $k_i\neq0$ for an $i\in[n]$. Note that $\tilde{q}_{i}>0$ and $p_{i}\neq0$, thus $\tilde{q}_{i}p_{i}k_{i}\neq0$. We obtain a dependent vector for $x_1,\ldots,x_n$.
\end{proof}

In the proof of Proposition \ref{scale}, from a dependent vector for any $x_1^{{p_1}/{q_1}},x_2^{{p_2}/{q_2}},\ldots, x_n^{{p_n}/{q_n}}$, one recovers a dependent vector for the original numbers $x_1,x_2,\ldots,x_n$. In particular, one recovers a dependent vector for $\{x_i\}_{i\in P}$ (defined in (\ref{reducedequation})) from a dependent vector for $\textbf{R}(x_{i_1}),\textbf{R}(x_{i_2}),\ldots,\textbf{R}(x_{i_s})$, which are all rational numbers. Computing multiplicatively dependent vectors for rational numbers is equivalent to solving linear Diophantine equations. We show this in the following by an example.
\begin{example}
\label{rationalcase}
Set $x_1=\frac{21}{4},x_2=\frac{27}{50},x_3=\frac{245}{32},x_4=\frac{16}{7}$. We solve the equation $x_1^{k_1}x_2^{k_2}x_3^{k_3}x_4^{k_4}=1$ with unknown integers $k_i$. One factors those rationals and collects factors sharing a same base:
\begin{equation}\begin{array}{l}
\quad1=x_1^{k_1}x_2^{k_2}x_3^{k_3}x_4^{k_4}\\
\\
\quad\;\;\;=\Big(\frac{21}{4}\Big)^{k_1}\cdot\Big(\frac{27}{50}\Big)^{k_2}\cdot\Big(\frac{245}{32}\Big)^{k_3}\cdot\Big(\frac{16}{7}\Big)^{k_4}\\
\\
\quad\;\;\;=\Big(\frac{3\cdot7}{2^2}\Big)^{k_1}\cdot\Big(\frac{3^3}{2\cdot5^2}\Big)^{k_2}\cdot\Big(\frac{5\cdot7^2}{2^5}\Big)^{k_3}\cdot\Big(\frac{2^4}{7}\Big)^{k_4}\\
\\
\quad\;\;\;=2^{-2k_1-k_2-5k_3+4k_4}\cdot3^{k_1+3k_2}\cdot5^{-2k_2+k_3}\cdot7^{k_1+2k_3-k_4}.
\end{array}\label{above}\end{equation}
By the fundamental theorem of arithmetic, any vector $(l_1,\ldots,l_4)^T$ in $\mathbb{Z}^4\backslash\{\mathbf{0}\}$ satisfies  $2^{l_1}3^{l_2}5^{l_3}7^{l_4}\neq1$. Hence (\ref{above}) is equivalent to
\[\left\{ {\begin{array}{*{20}{c}}
-2k_1&-&k_2&-&5k_3&+&4k_4&=&0,\\
k_1&+&3k_2&&&&&=&0,\\
&-&2k_2&+&k_3&&&=&0,\\
k_1&&&+&2k_3&-&k_4&=&0,
\end{array}} \right.\]
which has a unique solution $\mathbf{0}$. Hence $x_1,x_2,x_3,x_4$ are multiplicatively independent.
\end{example}

\section{Main algorithms}\label{sec:alg}

\subsection{Preprocess for Algebraic Numbers}
There are three ways to reduce the degrees of the algebraic numbers that we are going to deal with.  First, we use a function  \textbf{RootOfUnityTest} to decide whether a nonzero algebraic number is a root of unity and return its order if it is. Second, \textbf{RootOfRationalTest} in Algorithm \ref{algror} is a function to decide whether a nonzero algebraic number is a root of rational and return its rational order. Finally, the function \textbf{DegreeReduction} in Algorithm \ref{algdegred} is developed in order that for a given $\mathfrak{a}\in\overline{\mathbb{Q}}^*$, one finds an integer $q \geq 1$ \emph{s.t.} $\textbf{deg}(\mathfrak{a}^q)=\min_{k\in \mathbb{Z},k\geq1}\textbf{deg}(\mathfrak{a}^k)$ and the minimal polynomial of $\mathfrak{a}^q$.
\begin{definition}
\label{reddegree}
The \emph{reduced degree} of $\mathfrak{a}\in\overline{\mathbb{Q}}^*$  is defined by
$
\textbf{\emph{rdeg}}(\mathfrak{a})$\\$=\min_{k\in \mathbb{Z},k\geq1}\textbf{\emph{deg}}(\mathfrak{a}^k).
$
\emph{The set of reducing exponents} of $\mathfrak{a}$ is given by
 $
\textbf{\emph{Rexp}}(\mathfrak{a})=\{q\in\mathbb{Z}^*|\textbf{\emph{deg}}(\mathfrak{a}^q)=\textbf{\emph{rdeg}}(\mathfrak{a})\}
$. We say $\mathfrak{a}$ is \emph{degree reducible} if  $\textbf{\emph{rdeg}}(\mathfrak{a})<\textbf{\emph{deg}}(\mathfrak{a})$, otherwise \emph{degree irreducible}.
\end{definition}
\begin{example}
Set $\mathfrak{a}=(\sqrt{5}-2)e^{\frac{\pi \sqrt{-1}}{3}}$ with minimal polynomial $f(t)=t^4 + 4 t^3 + 17 t^2 - 4 t + 1$. Then \textbf{\emph{deg}}$(\mathfrak{a}^3)=2<4=\textbf{\emph {deg}}(\mathfrak{a})$. Hence $\mathfrak{a}$ is degree reducible.
\end{example}
We permute algebraic numbers $x_1,x_2,\ldots,x_n$ as follows
\begin{equation}
\mathfrak{x}^T=(\mathfrak{x}_1,\ldots,\mathfrak{x}_n)=(\alpha_1,\ldots,\alpha_{r},\beta_1,\ldots,\beta_s,\gamma_1,\ldots,\gamma_m)
\label{rearrange}
\end{equation}
\emph{s.t.} $\{\alpha\}_{i=1}^r$ are roots of unity, $\{\beta\}_{i=1}^s$ are roots of rational but none of which is a root of unity and none of $\{\gamma\}_{i=1}^m$ is a root of rational. Suppose we find $p_i\in\textbf{Rexp}(\gamma_i)$, $p_i\geq1$ for $i\in[m]$. Then we deal with (reduced) algebraic numbers:
\begin{equation}
(y_1,\ldots,y_n)=(1_1,\ldots,1_r,\textbf{R}(\beta_1),\ldots,\textbf{R}(\beta_s),\gamma_1^{p_1}\cdots,\gamma_m^{p_m}).
\label{reduced}
\end{equation}
Set $\tilde{\gamma}_i=\gamma_i^{p_i}$, then none of $\{\tilde{\gamma}_i\}_{i=1}^m$ is a root of rational. And assume that $\{\tilde{\gamma}_i\}_{i=1}^t$, for some $t\leq m$, satisfy the non-degenerate condition in Definition \ref{nondegen}. Here we ensure the non-degenerate condition by checking the equality
$
\textbf{deg}(\tilde{\gamma_1}+\tilde{\gamma_2}+\cdots+\tilde{\gamma_t})=\prod_{i=1}^{t}\textbf{deg}(\tilde{\gamma_i}).
$ By Corollary \ref{suffcor}, $\{\tilde{\gamma}\}_{i=1}^t$ are multiplicatively independent. More generally, we have the following observation:
\begin{proposition}
\label{ind}
Suppose  $\{\delta_i\}_{i=1}^l\subset\{\textbf{\emph{R}}(\beta_i)\}_{i=1}^s$ is a multiplicatively independent sequence. Algebraic numbers $\tilde{\gamma}_1,\ldots,\tilde{\gamma}_t$  satisfy the non-degenerate condition, none of which is a root of rational. Then the sequence $\delta_1,\ldots,\delta_l,\tilde{\gamma}_1,\ldots,\tilde{\gamma}_t$ is multiplicatively independent.
\end{proposition}
\begin{proof}
Since $\delta_i$ are rational, $\delta_1,\ldots,\delta_l,\tilde{\gamma}_1,\ldots,\tilde{\gamma}_t$ satisfy the non-degenerate condition. Set
$
\delta_1^{k_1}\cdots \delta_l^{k_l}\tilde{\gamma}_1^{m_1}\cdots\tilde{\gamma}_t^{m_t}=1
$
with integer exponents. Then $\tilde{\gamma}_i^{m_i}\in\mathbb{Q}$ by Theorem \ref{suff}. Since $\tilde{\gamma}_i$ is not a root of rational, $m_i=0$ for $i\in[t]$. Then
$
\delta_1^{k_1}\cdots \delta_l^{k_l}=1
$
with $\delta_1,\ldots,\delta_l$ multiplicatively independent, so $k_i=0$ for $i\in[l]$.
\end{proof}
Algorithm \ref{getbasis} is the main algorithm of this paper, its key parts are Algorithm \ref{alggetpre-basis} (\textbf{GetPreBasis}) and Algorithm \ref{pre-basis2basis} (\textbf{PreBasis2Basis}). The function \textbf{Isomorphism} in Step 7 is given by Algorithm \ref{iso}.

\begin{algorithm}
\SetAlgoLined
\KwIn{Nonzero algebraic numbers $x=(x_1,x_2,\ldots,x_n)^T$.}
\KwOut{$\{Basis,I\}$;
$\;Basis$ is a basis of  $\mathcal{R}_x$ (Definition \ref{rx});\\
$\;I\subset [n]$ indexes a maximal independent sequence.
}
Preprocess $(x_1,x_2,\ldots,x_n)$ to obtain: $\mathfrak{x}$ in (\ref{rearrange})$, (y_i)_{i=1}^n$ in (\ref{reduced}), \begin{small}\big\{\textbf{Order}$(\alpha_i)\big\}_{i=1}^r$, $\big\{\textbf{Rorder}(\beta_i)\big\}_{i=1}^s$ \end{small}and \begin{small}$\big\{p_i\in\textbf{Rexp}(\gamma_i)\big\}_{i=1}^m$\end{small};\\
\label{getbasis}
$\{PreBasis,I\}=\textbf{ GetPreBasis                                }\;\;\;\;\;\;\;\;\;\;\;\;\;\;$
$\big((y_i)_{i=1}^n,\begin{small}\{\textbf{Order}(\alpha_i)\}_{i=1}^r,\{\textbf{Rorder}(\beta_i)\}_{i=1}^s,\end{small}\{p_i\}\begin{small}_{i=1}^m\end{small}\big);$\\
$J=[n]\backslash I$; \\
\textbf{if} $(J==\emptyset)$ \textbf{\{Return} $\{\emptyset,[n]\}$;\textbf{\} endif}\\
Set $J=\{j_1,j_2,\cdots\}$, $j_1<j_2<\cdots$, $PreBasis=\{w_{j_1},w_{j_2},\cdots\};$\\
$g=\text{GCD}(w_{j_1}(1),w_{j_1}(2),\ldots,w_{j_1}(j_1))$;\\
$a=\textbf{Isomorphism}(\mathfrak{x},w_{j_1}/g,g);$
$u_{j_1}=w_{j_1}/\text{GCD}(a,g);\;$(\ref{ini})\\
$Basis=\{u_{j_1}\}$; \\
\For{$(k=2,3,\ldots,|J|)$}
{
$u_{j_k}=\textbf{PreBasis2Basis}(w_{j_k},\mathfrak{x},u_{j_1},u_{j_2},\ldots,u_{j_{k-1}});$\\
$Basis=Basis\cup\{u_{j_k}\}$;\\
}
Since algebraic numbers are permuted in (\ref{rearrange}), re-express each $u_j$ and $I$ in the original indices;\\
\textbf{Return }$\{Basis, I\}$
\caption{\textbf{GetBasis}}
\end{algorithm}

\subsection{Constructing a Set of Pre-Basis Vectors}
In the following, all of $B_j,S_j,I_j,V_j,U_j,J,\tilde{V}_j,\textbf{min}_j,u_j,w_j$ are defined for $\mathfrak{x}$ in (\ref{rearrange}) instead of  $x=(x_1,\ldots,x_n)^T$. We introduce Algorithm \ref{alggetpre-basis} to compute $J$ (Definition \ref{bv} (ii)) and extract pre-basis vectors $w_j\in U_j$ (Definition \ref{prebv}) for each $j\in J$, as preparation for computing a basis $\{u_j\}_{j\in J}$. Moreover, the set $I$ in Algorithm \ref{alggetpre-basis} is successively equal to $I_0,I_1,\ldots,I_{n}$ defined in (\ref{Ij}) as the algorithm runs. Algorithm \ref{alggetpre-basis} deals with algebraic numbers in (\ref{reduced}) while returning dependent vectors $w_j$ for $\mathfrak{x}$ in (\ref{rearrange}). This is by recovering technique mentioned in the proof of Proposition \ref{scale}. Algorithm \ref{alggetpre-basis} works in the spirit of formula (\ref{Ij}). Moreover, a dependent vector $w_j$ is obtained if the multiplicative dependence condition holds in (\ref{Ij}). Steps 2-5 deal with those roots of unity, while Steps 6-15 process roots of rational. Step 16 adds those numbers of non-degenerate property to the set $I$ by Proposition \ref{ind}. Finally, Steps 17-26 cope with the rest numbers that are not root of unity and failed in the non-degenerate condition test. For designing the function \textbf{DecideDependence}, we need results in \cite{van1977multiplicative,loxton1983multiplicative,masser1988linear}.  Theorem 1 of \cite{van1977multiplicative},  Theorem 3 of \cite{loxton1983multiplicative} and Theorem $G_m$ of  \cite{masser1988linear}  are of the same form as follows:
\begin{theorem}
\label{bnd}
Set $x\in(\overline{\mathbb{Q}}^*)^n$ to be multiplicatively dependent, then there is an efficiently computable number $\emph{\textbf{bnd}}$ s.t. $\exists v\in \mathcal{R}_x\backslash\{\mathbf{0}\}$, $\|v\|_\infty\leq \emph{\textbf{bnd}}$.
\end{theorem}

The number \textbf{bnd} is given by corresponding theorems in \cite{van1977multiplicative,loxton1983multiplicative,masser1988linear}. Other similar theorems can be found in \cite{loher2004uniformly, tori, mat}. By definition $S_{j-1}$ is multiplicatively independent. If the sequence $S_{j-1},\mathfrak{x}_j$ is multiplicatively dependent, by Theorem \ref{bnd} it follows that $\exists v\in\mathbb{Z}^{I_{j-1}\cup\{j\}}\backslash\{\mathbf{0}\},$ \emph{s.t.} $\mathfrak{x}^v=1,\|v\|_\infty<\textbf{bnd} $ and $v(j)>0$. Function \textbf{DecideDependence} accepts a sequence of nonzero algebraic numbers $b_1,b_2,\ldots,b_\iota,b_{\iota+1}$, with$\{b_j\}_{j=1}^\iota$ multiplicatively independent, as its input. By exhaustive search in the box defined by \textbf{bnd}, \textbf{DecideDependence} returns $\{\textbf{True}, v\}$ if it finds a dependent vector $v$, $\{\textbf{False},\mathbf{0}\}$ otherwise. \textbf{DecideDependence} can be replaced by the function \textbf{FindRelations}  in \cite{kauers2005algorithms} $\S7.3$ which accepts an $x\in(\overline{\mathbb{Q}}^*)^n$ as its input and returns a basis of $\mathcal{R}_x$. In this case, it returns one basis vector or $\emptyset$.
\begin{algorithm}
\SetAlgoLined
\KwIn{Pre-processed algebraic numbers in (\ref{reduced});\\
\begin{small}
$\big\{\textbf{Order}(\alpha_i)\big\}_{i=1}^r, \big\{\textbf{Rorder}(\beta_i)\big\}_{i=1}^s, \big\{$\end{small}$p_i\;$\begin{small}$\in\textbf{Rexp}(\gamma_i)\big\}_{i=1}^m$\end{small} in (\ref{rearrange}).
}

\KwOut{$\{PreBasis,I\}$;\\
Pre-basis $\{w_j\}_{j\in J}$ (Definition \ref{prebv}) and $J$ (Definition \ref{bv} (ii));\\
$I\subset [n]$ indexes a maximal independent sequence of $\mathfrak{x}$.
}
\caption{\textbf{GetPreBasis}}
$J=\emptyset$;
$I=\emptyset$;
$j=1$;
$\text{$PreBasis$}=\emptyset$;\\
\For{$(j=1,2,\ldots,r)$}
{\label{alggetpre-basis}
$J=J\cup\{j\}$;
$w_j=(0_1,0_2,\ldots,0_{j-1},\begin{small}\textbf{Order}\end{small}(\alpha_j))^T$;\\
$\text{$PreBasis$}=\text{$PreBasis$}\cup\{w_j\}$;\\
}
\textbf{if }{$(s>0)$}\textbf{ \{}{$I=I\cup\{r+1\};$\textbf{\} endif}\\}
\For{$(j=r+2,r+3,\ldots,r+s)$}
{
Solve $\mathfrak{D}:\;\;\prod_{i\in I\cup \{j\}}y_i^{k_i}==1\wedge k_{j}>0,$ with unknown vector $(k_i)_{i\in I\cup \{j\}}$ by the method in Example \ref{rationalcase};\\
\textbf{if }$(\mathfrak{D}$ has no solutions $)$ \textbf{\{}$I=I\cup\{j\}$;\textbf{ Continue;\} endif}\\
$J=J\cup\{j\}$;\\
\For{$(i\in I\cup\{j\})$}
{
$w_j(i)=\begin{small}\textbf{Rorder}\end{small}(\beta_{i-r})\cdot w_j(i)$;\\
}
$\text{$PreBasis$}=\text{$PreBasis$}\cup\{w_j\}$;\\
}
$I=I\cup\{r+s+1,r+s+2,\ldots,r+s+t\}$;\\
\For{$(j=r+s+t+1,\ldots,n-1,n)$}
{
$\{Depnd,w_j\}=\textbf{DecideDependence}\big(\{y_i\}_{i\in I},y_j\big)$;\\
    \textbf{if }$(I=\emptyset\lor Depnd=\textbf{False})$\textbf{ \{}$I=I\cup\{j\}$;\textbf{ Continue;\} \textbf{endif}}\\
              $J=J\cup\{j\}$;\\
\For {$($each $i\in I\cup\{j\})$}
{
\textbf{if} ($i\leq r+s$) $\textbf{\{}w_j(i)=\begin{small}\textbf{Roder}(\beta_{i-r})\end{small}\cdot w_j(i);\textbf{\}}$\\
\textbf{else \{}$w_j(i)=p_{i-r-s}\cdot w_j(i);\textbf{\} endif}$
}
$\text{$PreBasis$}=\text{$PreBasis$}\cup\{w_j\}$;\\
}
\textbf{Return } $\{PreBasis,I\}$
\end{algorithm}
\subsection{Recovering Basis from Pre-Basis}
We recover a basis  $\{u_j\}_{j\in J}$ (Definition \ref{bv}) from the pre-basis $\{w_j\}_{j\in J}$ (Definition \ref{prebv}). Algorithm \ref{alggetpre-basis} obtains $J=\{j\in[n]\;|\;V_j\neq\emptyset\}=\{j\in[n]\;|\;U_j\neq\emptyset\}$ on which $u_j$ and $w_j$ are defined. We compute $u_j$ inductively. That is, if $J=\{j_1,j_2,\ldots,j_r\}$ then we compute $u_{j_2}$ based on $u_{j_1}$, compute $u_{j_3}$ based on $u_{j_1}$ and $u_{j_2}$, and so forth.

\subsubsection{Recovering  a Basis Vector}

To begin with, we need the first basis vector $u_{j_1}$. By Theorem \ref{basisthm}, $w_{j_1}=\tau u_{j_1}$ for an integer $\tau>0$. Set $g= \text{GCD}(w_{j_1}(1),w_{j_1}(2),$\\$\ldots,w_{j_1}(j_1))$.
Consider $\tilde{w}= w_{j_1}/g$, then $u_{j_1}=k \tilde{w}$ for an integer $k>0$. By the definition of $u_j$, $k= \min K$, where
 $
K=\{q\in\mathbb{Z}\;|\;q>0, \mathfrak{x}^{q\cdot\tilde{w}}=1\}.
$ Noting that $(\mathfrak{x}^{\tilde{w}})^g=\mathfrak{x}^{w_{j_1}}=1$, we compute by Algorithm \ref{iso} an integer $0\leq a<g$ \emph{s.t.} $\mathfrak{x}^{\tilde{w}}=e^{2a\pi\sqrt{-1}/g}$. Then $\mathfrak{x}^{q\cdot\tilde{w}}=1\Leftrightarrow
 g|q\cdot a.
$ Thus
\[\begin{array}{rcl}
K&=&\{q\in\mathbb{Z}\;|\;q>0,qa\text{ is a multiple of }g\}\\
&=&\{q\in\mathbb{Z}\;|\;q>0,qa\text{ is a common multiple of }a\text{ and }g\}.
\end{array}
\] Hence if $a>0$, $\min K={\text{LCM}(a,g)}/{a}={g}/{\text{GCD}(a,g)}$. If $a=0$ then $\mathfrak{x}^{\tilde{w}}=1$ and $\min K=1=g/\text{GCD}(a,g)$. In a nutshell
\begin{equation}
u_{j_1}={g}/{\text{GCD}(a,g)}\cdot \tilde{w}={w_{j_1}}/{\text{GCD}(a,g)}.
\label{ini}
\end{equation}

Suppose $u_{j_1},\ldots,u_{j_k}$ are obtained, we compute a vector in $\tilde{V}_{j_{k+1}}$ (Definition \ref{bv} (iv)) to be $u_{j_{k+1}}$. Set $\tau= {w_{j_{k+1}}(j_{k+1})}/{\textbf{min}_{j_{k+1}}}$$\in \mathbb{Z}$ and $\Lambda=\{\lambda\in[w_{j_{k+1}}(j_{k+1})]|\lambda$ divides $w_{j_{k+1}}(j_{k+1})\}.$ For each $\lambda\in\Lambda$, define an equation
\begin{equation}\textbf{E}_\lambda:\left\{ {\begin{array}{*{20}{c}}
v({j_{k + 1}}) &=& \frac{{{w_{{j_{k + 1}}}}({j_{k + 1}})}}{\lambda }&\;\;\text{(a)}\\
{{\bar w}_{{j_{k + 1}}}} - \lambda \bar v &=& \sum\limits_{\iota  = 1}^k {{q_\iota }} {u_{{j_\iota }}}&\;\;\text{(b)}\\
{\mathfrak{x}^v} &=& 1&\;\;\text{(c)}
\end{array}} \right.\label{E}\end{equation}
with $v\in\mathbb{Z}^{j_{k+1}}$ unknown vector and $q_\iota$ unknown integers. Here $\overline{v}= v|(j_{k+1}-1)$ and $\overline{w}_{j_{k+1}}= w_{j_{k+1}}|(j_{k+1}-1)$ are defined at the beginning of $\S$3.1.
\begin{proposition}
\label{equi}
For $\lambda\in\Lambda$, the following conditions are equivalent : \emph{(i)} there are integers $\{q_\iota\}_{\iota=1}^k$, s.t. $\{v,q_\iota\}$ is a solution to \textbf{E}$_\lambda$, \emph{(ii)} $v\in V_{j_{k+1}}$and $v(j_{k+1})=w_{j_{k+1}}$$(j_{k+1})/\lambda$.
\end{proposition}
\begin{proof}
(i) $\Rightarrow$ (ii) is trivial. Suppose (ii) holds, then so do (\ref{E}\textcolor{red}{a}) and (\ref{E}\textcolor{red}{c}). Since $(w_{j_{k+1}}-\lambda v)(j_{k+1})=0, (\overline{w}_{j_{k+1}}-\lambda \overline{v})\in \mathcal{R}_{\mathfrak{x}|(j_{k+1}-1)}=\textbf{span}\{u_{j_1},u_{j_2},\ldots,u_{j_k}\}$  by Theorem \ref{basisthm}. So (\ref{E}\textcolor{red}{b}) holds for some $q_1,q_2,\ldots,q_k\in\mathbb{Z}$, thus (i) follows.
\end{proof}
Consider Proposition \ref{equi} when $\lambda=\tau\in\Lambda$. Since $\tilde{V}_{j_{k+1}}\neq \emptyset$, any $v\in\tilde{V}_{j_{k+1}}$ satisfies condition (ii). Thus $v$ satisfies condition (i) as well. So $\textbf{E}_\tau$ has a solution. Define $\tilde{\Lambda}=\{\lambda\in\Lambda\;|\;\textbf{E}_\lambda \text{ has a solution}\}$ and $\lambda_m=\max \tilde{\Lambda}$, we indicate that $\tau=\lambda_m$.
\begin{proposition}
\label{maxlambda}
Set $\lambda_m=\max \tilde{\Lambda}$, then $\lambda_m=\tau=\frac{w_{j_{k+1}}(j_{k+1})}{\textbf{\emph{min}}_{j_{k+1}}}$. Each solution to \textbf{E}$_{\lambda_m}$ can be projected to a vector in $\tilde{V}_{j_{k+1}}$.
\end{proposition}
\begin{proof}
Since $\textbf{E}_{\lambda_m}$ has a solution, by Proposition \ref{equi}, $\exists v\in V_{j_{k+1}}$ \emph{s.t.}  $v(j_{k+1})=w_{j_{k+1}}(j_{k+1})/\lambda_m$. Then $\textbf{min}_{j_{k+1}}\leq w_{j_{k+1}}(j_{k+1})/\lambda_m$ by Definition \ref{bv} (iii). This is $\tau\geq\lambda_m$. Since $\textbf{E}_\tau$ has a solution, $\tau\in\tilde{\Lambda}$. Hence $\tau\leq\lambda_m$. Suppose $\{v,q_\iota\}$ is a solution to $\textbf{E}_{\lambda_m}$. By Proposition \ref{equi} its projection $v\in V_{j_{k+1}}$ and $v(j_{k+1})=w_{j_{k+1}}(j_{k+1})/\lambda_m=w_{j_{k+1}}(j_{k+1})/\tau=\textbf{min}_{j_{k+1}}$. Hence $v\in\tilde{V}_{j_{k+1}}$ by definition.
\end{proof}
In the spirit of Proposition \ref{maxlambda}, we arrange the numbers in $\Lambda$ decreasingly as: $\lambda^{(1)}>\lambda^{(2)}>\cdots.$ Then we solve a sequence of equations in the order: $\textbf{E}_{\lambda^{(1)}},\textbf{E}_{\lambda^{(2)}},\cdots.$ We stop as soon as some $\textbf{E}_{\lambda^{(i)}}$ has a solution (otherwise we move to the next one). Since $\textbf{E}_{\lambda^{(i)}}$ is the first set of equations to have a solution, $\lambda^{(i)}=\lambda_m$. Any solution to $\textbf{E}_{\lambda^{(i)}}$ can be projected to a vector in $\tilde{V}_{j_{k+1}}$, which we take as the value of $u_{j_{k+1}}$.
\subsubsection{Solving the Equation $\textbf{E}_\lambda$}

First, solve the linear Diophantine equation (\ref{E}\textcolor{red}{b}). If (\ref{E}\textcolor{red}{b}) has no solutions, then neither has $\textbf{E}_\lambda$. Suppose (\ref{E}\textcolor{red}{b}) has general solution:
\begin{equation}\left( {\begin{array}{*{20}{c}}
L\\
Q
\end{array}} \right) = \left( {\begin{array}{*{20}{c}}
{L_0}\\
{Q_0}
\end{array}} \right) + {z_1}\left( {\begin{array}{*{20}{c}}
{L_1}\\
{Q_1}
\end{array}} \right) +  \cdots  + {z_s}\left( {\begin{array}{*{20}{c}}
{L_s}\\
{Q_s}
\end{array}} \right)\label{generalsolution}\end{equation}
\noindent{}where $L_i\in \mathbb{Z}^{j_{k+1}-1}$, $Q_i\in\mathbb{Z}^{k}$ and $z_i$ are any integers. $L$ gives the value of $\overline{v}$ while $Q$ gives the value of the vector $(q_1,q_2,\ldots,q_k)^T$. Moreover,
\begin{equation}\left\{ {\begin{array}{*{20}{c}}
{{{\overline w}_{{j_{k + 1}}}} - \lambda {L_0} = \sum\limits_{\iota  = 1}^k {{Q_0}(\iota )} {u_{{j_\iota }}}},\\
{ - \lambda {L_i} = \sum\limits_{\iota  = 1}^k {{Q_i}(\iota )} {u_{{j_\iota }}},\;{\rm{for}}\;i \in[s].}
\end{array}} \label{solu}\right.\end{equation}
\noindent{}That is, $(L_0^T,Q_0^T)^T$ is a solution to (\ref{E}\textcolor{red}{b}) and $(L_i^T,M_i^T)^T$ form a basis of the solution lattice to the homogeneous version of (\ref{E}\textcolor{red}{b}).

We concern ourselves with the problem whether there are integers $z_1,z_2,\ldots,z_s$ \emph{s.t.} while taking $v=(L^T,{w_{j_{k+1}}(j_{k+1})}/{\lambda})^T$ as in (\ref{generalsolution}), (\ref{E}\textcolor{red}{c}) holds. If there are, then $\textbf{E}_\lambda$  has a solution, otherwise it has not. In fact
\[\begin{array}{rcl}
\mathfrak{x}^v&=&{\mathfrak{x}^{{L_0} + {z_1}{L_1} +  \cdots  + {z_s}{L_s}}} \cdot \mathfrak{x}_{j_{k + 1}}^{{{{w_{{j_{k + 1}}}}({j_{k + 1}})}}/{\lambda }}\\
&=& \mathfrak{x}^{z_1L_1} \cdots \mathfrak{x}^{z_sL_s}\cdot (\mathfrak{x}^{L_0}\mathfrak{x}_{j_{k + 1}}^{ w_{j_{k + 1}}(j_{k + 1})/\lambda}).
\end{array}
\]
 Since $\mathfrak{x}^{u_{j_\iota}}=1$, it follows from the 1$^\text{st}$ equation of (\ref{solu}) that $\mathfrak{x}^{\lambda {L_0}}=\mathfrak{x}^{{{\overline w}_{{j_{k + 1}}}}}$. Moreover, $\mathfrak{x}^{{{\overline w}_{{j_{k + 1}}}}}\cdot \mathfrak{x}_{j_{k+1}}^{w_{j_{k+1}}(j_{k+1})}=1$ by definition, one obtains
\[\begin{array}{rl}
&{\big({{\mathfrak{x}^{{L_0}}}\mathfrak{x}_{j_{k + 1}}^{{w_{{j_{k + 1}}}}({j_{k + 1}})/\lambda }}\big)^\lambda }\\
= &{\mathfrak{x}^{\lambda {L_0}}}\mathfrak{x}_{j_{k + 1}}^{{w_{{j_{k + 1}}}}({j_{k + 1}})}\\
 =& {\mathfrak{x}^{{{\overline w}_{{j_{k + 1}}}}}}\mathfrak{x}_{j_{k + 1}}^{{w_{{j_{k + 1}}}}({j_{k + 1}})}\\
= &1.
\end{array}
\]
For $L_i,i\in[s]$, using the 2$^\text{nd}$ equation of (\ref{solu}), we have ${(\mathfrak{x}^{L_i})}^{\lambda}=\mathfrak{x}^{\lambda L_i}=1$. Now  we observe that \begin{equation}\Gamma_0=\mathfrak{x}^{L_0}\mathfrak{x}^{ w_{j_{k + 1}}(j_{k + 1})/\lambda}_{j_{k + 1}}\;\text{  and  }\;\Gamma_i= \mathfrak{x}^{L_i}, i\in[s]\label{Gammai}\end{equation} are roots of the equation $\Gamma^\lambda-1=0$. By using Algorithm \ref{iso}, for each $\Gamma_i$, $i\in\{0\}$$\cup[s]$, one obtains an integer $0\leq a_i<\lambda$ \emph{s.t.} $\Gamma_i=e^{{2a_i\pi\sqrt{-1}}/{\lambda}}$. The set of all roots of the equation $\Gamma^\lambda-1=0$ is isomorphic to the group $\mathbb{Z}/\langle\lambda\rangle=\{0,1,\ldots,\lambda-1\}$ with addition modulo $\lambda$. The problem is reduced to whether there are integers $z_1,\ldots,z_s$ \emph{s.t.}
$\Gamma_1^{z_1}\Gamma_2^{z_2}\cdots \Gamma_s^{z_s}\Gamma_0=1$,
 which is equivalent to $\exists p\in\mathbb{Z}$,
\begin{equation} a_1z_1+\cdots +a_sz_s+a_0=p\lambda.\label{single} \end{equation}
This equation with $z_1,\ldots,z_s, p$ unknown integers can be efficiently solved. If (\ref{single}) has no solutions, neither has $\textbf{E}_\lambda$. Otherwise we get the values of $z_i$ and obtain $L$ and $Q$ from (\ref{generalsolution}). Then $\{v=(L^T,{w_{j_{k+1}}(j_{k+1})}/{\lambda})^T, (q_1,\ldots,q_k)=Q^T\}$ is a solution to $\textbf{E}_\lambda$. Finally, assign $u_{j_{k+1}}$$=(L^T,{w_{j_{k+1}}(j_{k+1})}/{\lambda})^T$. We summarize all these by Algorithm \ref{pre-basis2basis}.

\begin{algorithm}
\SetAlgoLined
\KwIn{Pre-basis vector $w_{j_{k+1}}$; $\mathfrak{x}$ in (\ref{rearrange});\\ \quad\quad\quad Basis vectors: $u_{j_1},u_{j_2},\ldots,u_{j_k}$.}
\KwOut{Next basis vector $u_{j_{k+1}}$.}
\caption{\textbf{PreBasis2Basis}}
Let $\lambda^{(1)}=w_{j_{k+1}}(j_{k+1})>\lambda^{(2)}>\cdots>\lambda^{(\gamma)}=1$ be all positive numbers that divide $w_{j_{k+1}}(j_{k+1})$;\\
\For{$(\lambda=\lambda^{(1)},\lambda^{(2)},\ldots,1)$}
{
     Solve (\ref{E}\textcolor{red}{b}) for $L_0,L_1,\ldots,L_s$;\\
\label{pre-basis2basis}
    \textbf{if }$($(\ref{E}\textcolor{red}{b}) has no solutions$)$\textbf{ \{Continue};\textbf{\} endif}\\
$a_0=\textbf{Isomorphism}\Big(\mathfrak{x},\big(L_0^T,{w_{j_{k+1}}(j_{k+1})}/{\lambda}\big)^T,\lambda\Big);$\\
\For{$(i=1,\ldots,s)$}
{
$a_i=\textbf{Isomorphism}(\mathfrak{x},L_i,\lambda)$;
}
Solve (\ref{single}) for $z_1,z_2,\ldots,z_s,p$;\\
                  \textbf{if }$((\ref{single})$ has no solutions$)$\textbf{ \{Continue};\textbf{\} endif}\\
                  \textbf{Return }$u_{j_{k+1}}=(L^T,{w_{j_{k+1}}(j_{k+1})}/{\lambda})^T$;\\
}
\end{algorithm}

Denote a rectangle by
$
[a,b;c,d]=\{z\in\mathbb{C}\;|\;a\leq\Re(z)\leq b,c\leq\Im(z)\leq d\},
$
 call $a,b,c,d$ the coordinates of the rectangle $[a,b;c,d]$. By complex root isolation in \cite{isolation}, we isolate each $x_j$ by small rectangles $R_j$ on the complex plane with rational coordinates. Assume in addition that $\forall j,0\not\in R_j$. For each $R_j$, efficiently chose an interval $\theta_j=[\underline\theta_j,\overline\theta_j]$ with $\underline\theta_j\leq\overline\theta_j$ rational numbers, \emph{s.t.} $R_j\subset\{\rho e^{\pi\eta\sqrt{-1}}\in\mathbb{C}|\rho>0,\;\underline\theta_j\leq\eta\leq\overline\theta_j\}$.  That can be done by standard method such as series expansion. Using interval arithmetic, we provide Algorithm \ref{iso} for implementing the group isomorphism $\{\Gamma\;|\;\Gamma^\lambda=1\}\rightarrow\mathbb{Z}/\langle\lambda\rangle$. In Step 2, the condition ``$\exists a,b\in\mathbb{Z}\cap[0,\lambda), a\neq b$ \emph{s.t.} $\{2a/\lambda,2b/\lambda\}\subset\Theta$''  is equivalent to ``the  length of  $\Theta\geq 2/\lambda$''. The number ``50\%'' makes sure that the length of each $\theta_j$ converges to $0$ (hence so does the length of $\Theta$) so that  the algorithm terminates.

\begin{algorithm}
\SetAlgoLined
\KwIn{$x=(x_1,\ldots,x_n)^T$$\in(\overline{\mathbb{Q}}^*)^n$;\\\quad\quad$\;\;\;$ A vector $v\in \mathbb{Z}^n$;  an integer $\lambda>0$ \emph{s.t.} $x^{\lambda v}=1$.}
\KwOut{An integer $0\leq a<\lambda$ \emph{s.t.} $x^v=e^{2a\pi\sqrt{-1}/\lambda}$.}
\caption{\textbf{Isomorphism}:  $\{\Gamma\;|\;\Gamma^\lambda=1\}\rightarrow\mathbb{Z}/\langle\lambda\rangle$}
$\Theta=\theta_1=\cdots=\theta_n=[0,2]$;\\
\While {$(\exists a,b\in\mathbb{Z}\cap[0,\lambda), a\neq b$ \emph{s.t.} $\{2a/\lambda,2b/\lambda\}\subset\Theta)$}
{
Isolate each $x_j$ by smaller rectangle $R_j$ and re-compute an interval $\theta_j$ with length decreases by at least $50\%$;\\
      $\Theta=v(1){\theta_1}+v(2){\theta_2}+\cdots +v(n){\theta_n}$;\\
\label{iso}
}
\textbf{Return} the only integer $a$  so that $2a/\lambda\in\Theta$;
\end{algorithm}
\section{Degree Reduction Functions in Preprocessing}\label{sec:red}
\subsection{Recognizing Roots of Rational}
\begin{proposition}
\label{necess}
For $p(t)=t^d+a_{d-1}t^{d-1}+\cdots+a_0\in\mathbb{R}[t]$ whose roots $z_1,\ldots,z_d$ are all of modulus $1$, $p(t)=\pm t^dp(1/t)$.
\end{proposition}
\begin{proof}
Define $\tilde{p}(t)= t^dp(1/t)$. Roots of $\tilde{p}(t)$ are $(z_1^{-1},\ldots,z_d^{-1})=$\\$(\overline{z}_1,\ldots,\overline{z}_d)$. This is a re-permutation of $(z_1,\ldots,z_d)$. Thus $p(t)=$\\$c\cdot\tilde{p}(t), c\in\mathbb{C}\backslash\{0\}$. In fact  $c=1/a_0$, and $|a_0|=\prod_{i\in[d]}|z_i|=1.$
\end{proof}
\begin{proposition}
\label{samerorder}
For irreducible $p(t)\in\mathbb{Q}[t]$ with roots $z_1,\ldots,z_d$ $\in\mathbb{C}$, exactly one of the following conditions holds: (i) none of $z_1,\ldots,z_d$ is a root of rational; (ii) all of $z_1,\ldots,z_d$ are roots of rational, and $\textbf{\emph{R}}(z_1)=\textbf{\emph{R}}(z_j)$, $\textbf{\emph{Rorder}}(z_1)=\textbf{\emph{Rorder}}(z_j)$, $2\leq j\leq d$.
\end{proposition}
\begin{proof}
Suppose $z_1$ is a root of rational of rational order $k$ and $\textbf{R}(z_1)$\\$=r$. Then $z_1^k-r=0$ and hence $p(t)|t^k-r$. Thus  $p(z_j)=0$ implies $z_j^k-r=0$, $j=2,\ldots,d$. It follows that $\textbf{Rorder}(z_j)|k$. Exchanging the roles of $z_1$ and $z_j$, one sees that $k|\textbf{Rorder}(z_j)$. Hence $k=\textbf{Rorder}(z_j)$ and $\textbf{R}(z_j)=r$.
\end{proof}
Set $\mathfrak{a}$ to be a root of rational, $p(t)$ its monic minimal polynomial with $p(0)=a_0$. All complex roots of $p(t)$ are of the same modulus, say $\ell$, by Proposition \ref{samerorder}, then $|a_0|=\ell^d$. Define monic polynomial
$
\overline{p}(t)= p(\sqrt[d]{|a_0|}\cdot t)/|a_0|
\in \mathbb{R}[t]$ whose roots are all of modulus $1$. Then if the conclusion of  Proposition \ref{necess} fails to hold for $\overline{p}(t)$, one concludes that $\mathfrak{a}$ is not a root of rational.
\begin{proposition}
\label{rorpro}
(Lemma 3.5 in \cite{lemma}) Let $F$ be a field, $E=F(\alpha)$, $[E:F]=d$ and $\alpha^m\in F$. If $f(t)$ is the monic irreducible polynomial of $\alpha$ over $F$, then $\zeta \alpha^d=(-1)^d f(0)\in F$ for some $\zeta $ s.t. $ \zeta^m=1$.
\end{proposition}
Set $\alpha=\mathfrak{a}, F=\mathbb{Q}, E=\mathbb{Q}(\mathfrak{a})$ and  $m=\textbf{Rorder}(\mathfrak{a})$, then $\zeta =(-1)^da_0/\mathfrak{a}^d$. By Proposition \ref{rorpro}, $\zeta^{\textbf{Rorder}(\mathfrak{a})}=1.$ Thus $\exists k_1\in\mathbb{Z}$ \text{ \emph{s.t.} (i)} $k_1\cdot\textbf{Order}(\zeta)=\textbf{Rorder}(\mathfrak{a}).$ Note that $\zeta^{\textbf{Order}(\zeta)}=1$ means $\mathfrak{a}^{d\cdot\textbf{Order}(\zeta)}=((-1)^da_0)^{\textbf{Order}(\zeta)}\in\mathbb{Q}$. This implies $\exists k_2\in \mathbb{Z}$ \emph{s.t.} (ii)
 $
k_2\cdot\textbf{Rorder}(\mathfrak{a})=d\cdot\textbf{Order}(\zeta).
$
Combining (i) and (ii) we have $k_1\cdot k_2=d.$ These lead to Algorithm \ref{algror} deciding whether an algebraic number, given its minimal polynomial, is a root of rational.
\begin{algorithm}
\SetAlgoLined
\KwIn{Irreducible monic $p(t)\in\mathbb{Q}[t]\text{ \emph{s.t.} }p(0)=a_0,p(\mathfrak{a})=0$.}
\KwOut{$\{\textbf{Rorder}(\mathfrak{a})$, \textbf{R}$(\mathfrak{a})\}$ (Definition \ref{rorn}).}
Compute $\overline{p}(t)=p(\sqrt[d]{|a_0|}\cdot t)/|a_0|$;\\
\textbf{if} $(\overline{p}(t)\neq \textbf{sgn}(a_0)\cdot t^d\overline{p}(1/t))$\textbf{ \{Return} $\{0,1\}$;\textbf{\} endif}\\
$ f(t)=\textbf{MinimalPolynomial}(\mathfrak{a}^d/((-1)^da_0))$;\\
$\{RootOfUnity, order\}=\textbf{RootOfUnityTest}(f(t))$;\\
\textbf{if} (RootOfUnity$==\textbf{False}$)\textbf{ \{Return} $\{0,1\}$;\textbf{\} endif}\\
\label{algror}
Arrange positive divisors of $d$ increasingly: $\lambda_1<\lambda_2<\cdots$;\\
\For{$(\lambda=\lambda_1,\lambda_2,\cdots)$}
{
$r(t)=\textbf{PolynomialRemainder}(t^{\lambda\cdot order},p(t))$;\\
\textbf{if} $(r(t) \in \mathbb{Q})$ \textbf{\{Return} $\{\lambda\cdot order, r(t)\}$;\textbf{\} endif}\\
}
\caption{\textbf{RootOfRationalTest}}
\end{algorithm}
\subsection{Computing Reduced Degree}
Set $\mathfrak{a}$ to be an algebraic number with minimal polynomial $p(t)$ whose complex roots are $\mathfrak{a}=z_1,z_2,\ldots,z_d$.
\begin{proposition}
For integer $m\geq 1$, if the minimal polynomial of $\mathfrak{a}^m$ is $f(t)$, then $\{z\in\mathbb{C}|f(z)=0\}=\{z_i^m\;|\;i\in[d]\}$.
\end{proposition}
\begin{proof}
Define $\mathbb{E}=\mathbb{Q}(\mathfrak{a},z_2,\ldots,z_d)$, $G= \text{Gal}(\mathbb{E}/\mathbb{Q})$ and $G(\mathfrak{a}^m)=\{\sigma(\mathfrak{a}^m)|\sigma\in G\}=\{\beta_1=\mathfrak{a}^m,\beta_2,\ldots,\beta_r\}$. Then it is not hard to see that  $f(t)=c\cdot\prod_{i=1}^{r}(t-\beta_i)$ for some $c\in\mathbb{Q}^*$. Hence $\{z\in\mathbb{C}|f(z)=0\}=\{\sigma(\mathfrak{a}^m)|\sigma\in G\}=\{\big(\sigma(\mathfrak{a})\big)^m|\sigma\in G\}=\{z_i^m\;|\;i\in[d]\}$. The last equality holds since the Galois group of irreducible polynomial $p(t)$ operates transitively on its roots.
\end{proof}
Set $ R_m=\{z_i^m\;|\;i\in[d]\}$ to be complex roots of the minimal polynomial of $\mathfrak{a}^m$, then $|R_m|=\textbf{deg}(\mathfrak{a}^m)$. By Definition \ref{reddegree}, \text{(i) }$\textbf{rdeg}(\mathfrak{a})$\\$=\min_{m\in\mathbb{Z},m\geq1}|R_m|$.
\begin{proposition}
If $z_i/z_j$ is a root of unity for some $1\leq i\neq j\leq d$, then $\mathfrak{a}$ is degree reducible.
\end{proposition}
\begin{proof}
Set $m=\mathbf{Order}(z_i/z_j)$, then $\mathbf{deg}(\alpha^m)=|R_m|<|R_1|=d=\mathbf{deg}(\alpha)$.
\end{proof}
Define an  equivalent relation $\sim$ on $\mathbb{C}$ by: $z\sim z'$ iff $z/z'$ is a root of unity. Because $z_i^m= z_j^m$ implies $z_i\sim z_j$, \text{(ii) }$|R_1/$$\sim$$|\leq|R_m|$ for any integer $m\geq1$. Denote by $M$ the unitary order of $p(t)$ (\cite{yokoyama1995finding} Definition 2.3), then $\forall i,j\in[d], (z_i\sim z_j \Leftrightarrow z_i^M=z_j^M)$. Thus $|R_M|=|R_1/$$\sim$$|$. Combine this with (i) and (ii), it follows that \text{(iii) }$|R_1/$$\sim$$|=\min_{m\in\mathbb{Z},m\geq1}|R_m|=\textbf{rdeg}(\mathfrak{a}).$ Since $z_i\sim z_j$ iff $z_i^m\sim z_j^m$, we conclude that  $|R_1/$$\sim$$|=|R_m/$$\sim$$|$.
\begin{proposition}
\label{correct}
$|R_m/$$\sim$$|=|R_m|\Rightarrow \textbf{\emph{rdeg}}(\mathfrak{a})=|R_m|=\textbf{\emph{deg}}(\mathfrak{a}^m)$.
\end{proposition}
This is obvious since $|R_m|=|R_m/$$\sim$$|=|R_1/$$\sim$$|$ and (iii) holds.  \cite{yokoyama1995finding} provides an efficient algorithm \textbf{Unitary-Test} so that for $p(x)\in \mathbb{Q}[x]$, if the quotient of  a pair of  roots of $p$ is a root of unity, it returns $\{\textbf{True},k\}$, where $k$ is the order of the quotient.  Otherwise it returns $\{\textbf{False}, 0\}$. Based on this, we develop Algorithm \ref{algdegred} to compute a reducing exponent and the reduced degree of a given algebraic number. The algorithm terminates since the number of roots of $f(t)$ decreases in Step 5. Its correctness is justified by Proposition \ref{correct}.
\begin{algorithm}
\SetAlgoLined
\KwIn{Irreducible polynomial $p(t)\in\mathbb{Q}[t]$ with a root $\mathfrak{a}$.}
\KwOut{$\{prod, f(t)\}$; An integer $prod\geq1$, $prod\in \textbf{Rexp}(\mathfrak{a})$\\(Definition \ref{reddegree}); the minimal polynomial $f(t)$ of  $\mathfrak{a}^{prod}.$
}
$prod=1;$
$f(t)=p(t)$; \\
Suppose $f(t)$ has a root $r$;\\
$\{Reducible,k\}=\textbf{Unitary-Test}(f(t))$;\\
\textbf{if }($Reducible==\textbf{False}$) \textbf{\{Return} $\{prod,f(t)\}$;\textbf{\} endif}\\
$prod=prod\cdot k;$
$f(t)=\textbf{MinimalPolynomial }(r^k)$;\\
\textbf{Goto} Step 2;\label{algdegred}
\caption{\textbf{DegreeReduction}}
\end{algorithm}
\section{Experiments and Application}
\label{sec:example}
We show experimental results\footnotemark[1]\footnotetext[1]{For details turn to: \textcolor{blue}{https://github.com/zt007/exm/blob/master/details}.} verifying the effectiveness of the framework to tackle problems larger than those that \textbf{FindRelations} can handle. The Mathematica package \textbf{FindRelations} is available at\[\text{\textcolor{blue}{https://www3.risc.jku.at/research/combinat/software/}.}\]We implemented our algorithms with Mathematica. All results are obtained on a laptop of WINDOWS 7 SYSTEM with 4GB RAM and a 2.53GHz Intel Core i3 processor with 4 cores.

\begin{center}
\renewcommand\arraystretch{0.93}
\begin{tabular}	{c|c|c|c|c}
\hline
NO.&TD&RTD&R/B&Time (s) FD/FF/GE\\
\hline
1&9.3$\times10^6$&1.5$\times10^6$&4/3&26/2196/>3600\\
\hline
2&1.2$\times10^5$&336&5/2&114.8/114.3/>3600\\
\hline
3a&30&30&3/0&0.035/0.032/8.9\\
3b&120&120&4/0&2.69/2.68/>3600\\
3c&600&600&5/0&360/361/>3600\\
\hline
4a&108&108&4/1&11.18/0.79/0.11\\
4b&256&256&4/0&>3600/188.4/78.6\\
4c&729&729&4/2&>3600/2049/>3600\\
\hline
5a&$4.7\times10^4$&1&1/5&0.433/0.433/2.227\\
5b&$6.6\times10^7$&1&4/2&0.712/0.720/>3600\\
5c&$3.8\times10^{27}$&1&11/9&203.98/203.77/>3600\\
\hline
\end{tabular}
\end{center}
In the table TD stands for the product of degrees of given algebraic numbers, RTD the product of reduced degrees, R the rank (Definition \ref{algerank}) and B cardinality of the basis. FD and  FF are respectively the main algorithm and the main algorithm with \textbf{DecideDependence} replaced by \textbf{FindRelations}, while GE means applying \textbf{FindRelations} directly to the problem. TD (RTD) is the upper bound of the degree of the extended field generated by (reduced) given algebraic numbers. It indicates the difficulty for algorithm GE (FD, FF) to compute in the extended field. For algorithm FD (FF), the number of input algebraic numbers for function \textbf{DecideDependence} (\textbf{FindRelations}) is at most R$+1$ throughout the computation, thus R also characterizes the computational difficulty.

These three algorithms (FD, FF, GE) are all very efficient for small inputs (\emph{e.g.}, roots of an irreducible polynomial of degree 3) which we leave out in the table. The runtimes of these algorithms depend intensively on the input. We see that these algorithms perform differently from each other in the examples. In Example 1 both TD and RTD are large. GE fails to give an answer within an hour, while FD (FF) deals with partial input by \textbf{DecideDependence} (\textbf{FindRelations}) and is faster. In Example 2, RTD is much smaller than TD, FD (FF) is more efficient than GE which does not contain degree reduction. The non-degenerate condition holds for Examples 3a-c. The certificate for multiplicative independence is effective in this case. In Examples 4a-b GE does better than FD (FF) since the degree reduction does not help, R$/$(R$+$B) is close to $1$ and the non-degenerate condition fails to hold. However FF still gives correct answer. FF outperforms GE in Example 4c since R/(R$+$B)  is relatively small. Examples 5a-c deal with the case where all input numbers are roots of rationals. FD (FF) can handle very large problem of this type. To conclude, (i) RTD $<<$ TD, (ii) R/(R$+$B) $<<$ 1 and (iii) the non-degenerate condition holds for more input numbers are good for FD (FF) to tackle problems larger than those that GE can handle.

An interesting application is to compute the invariant polynomial ideal of  linear loops considered in \cite{lvov2010polynomial, structure} of the form:
\[X= b; {\rm While} \textbf{ True  }{\rm do} \;X= AX;\] Here $b\in\mathbb{Q}^m$, $A\in\mathbb{Q}^{m\times m}$ is diagonalizable with nonzero eigenvalues. Let $X= (X_1,\ldots,X_m)^T$ be the vector of indeterminates and $I(A,b)= \{f(X)\in\mathbb{C}[X]\;|\;f(A^kb)=0,\forall k\geq0\}$ the \emph{invariant polynomial ideal}. Suppose $A^T$$=$$ PDP^{-1}$, $D=\textbf{diag}\{x_1,x_2,\ldots,x_m\}$,  $x= (x_1,x_2,\ldots,x_m)^T$ and $\tilde{b}= P^Tb\in(\mathbb{C}^*)^m$. For $v\in\mathcal{R}_x$, define $v_+=(\max\{v(1),0\},\ldots,\max\{v(m),0\})^T$ and $v_-= (-v)_+$. Then by Theorems 3 and 2 in \cite{structure} one concludes
\begin{equation}
I(A,b)=\langle\{({\tilde{b}}^{v_-})(P^TX)^{v_+}-({\tilde{b}}^{v_+})(P^TX)^{v_-}|v\in\mathcal{R}_x\}\rangle.
\label{thm3}
\end{equation}
Define
$\mathfrak{g}_v(X)=({\tilde{b}}^{v_-})(P^TX)^{v_+}-({\tilde{b}}^{v_+})(P^TX)^{v_-}$,
we observe that $\forall k\in\mathbb{Z}$, $\mathfrak{g}_v(X)|\mathfrak{g}_{kv}(X)$. Hence if $\textbf{rank}(\mathcal{R}_x)=1$ with only one basis vector $u$, (\ref{thm3}) becomes
$
I(A,b)=\langle\mathfrak{g}_u(X)\rangle.
$
 Set $p_i(X) \in \mathbb{C}[X],i\in[s]$, consider the loop:
\begin{equation}\begin{array}{l}
 X= b;\\
While\;(p_1(X)==0\land\cdots\land p_s(X)==0)\\
do\;(X= AX;)
\end{array}\label{loop}\end{equation}
which does not terminate iff $\mathfrak{g}_u(X)|p_i(X),i\in[s]$. Set \[A = \left( {\begin{array}{*{20}{c}}
4&{226}&2&1&{ - 117}\\
1&{126}&1&0&{ - 64}\\
0&{ - 91}&0&{ - 1}&{46}\\
0&{80}&1&0&{ - 40}\\
4&{232}&2&1&{ - 120}
\end{array}} \right),b = \left( {\begin{array}{*{20}{c}}
2\\
{ - 1}\\
0\\
3\\
{ - 4}
\end{array}} \right).\] The only basis vector of the exponent  lattice defined by eigenvalues of matrix $A$ is $u=(0, 1, 1, 0, 1)^T$. Computing by definition shows
\\
\begin{small}$
\\\mathfrak{g}_u(X)=
7674169 X_1^3 - 31858655 X_1^2 X_2 + 22396826 X_1 X_2^2 - 165997 X_4^3- $\\$
255380 X_2 X_3^2  + 12769 X_3^3 +153228 X_4 X_5^2 + 1442897 X_5^3+459684 X_1 X_3^2 - $\\$
3051791 X_1 X_2 X_4 + 3639165 X_1^2 X_3 +8504154 X_2 X_3 X_5- 472453 X_2^2 X_3 -$\\$
5694974 X_1 X_4 X_5 +1391821 X_2 X_3 X_4- 127690 X_1 X_4^2 + 2106885 X_2 X_4^2-$\\$
2311189 X_1 X_3 X_5-1442897 X_3 X_4 X_5 - 906599 X_4^2 X_5 -3639165 X_1 X_5^2 - $\\$
204304 X_3 X_4^2  - 5465132 X_1^2 X_5 + 35817045 X_1 X_2 X_5 + 3971159 X_1^2 X_4 - $\\$
2387803 X_2 X_5^2- 995982 X_3 X_5^2 -28347180 X_2^2 X_5 +9947051 X_2 X_4 X_5 -$\\$
12769 X_3^2 X_4 + 5528977 X_2^3 - 7380482 X_2^2 X_4  - 293687 X_3^2 X_5 -$\\$
8899993 X_1 X_2 X_3  + 855523 X_1 X_3 X_4- 858983399.
$
\end{small}
\\

The ideal in (\ref{thm3}) is a lattice ideal up to an invertible linear transformation in the coordinates. When there are at least two basis vectors, we may compute according to \cite{markov}, from the basis vectors, a Markov basis $\{\eta_j\}_{j=1}^r$, \emph{s.t.} $
I(A,b)=\langle\{\mathfrak{g}_{\eta_j}(X)\}_{j=1}^r\rangle
$
by Lemma A.1 of \cite{markov}. Then (\ref{loop}) does not terminate iff $p_i(X)\in\langle\{\mathfrak{g}_{\eta_j}(X)\}_{j=1}^r\rangle,i\in[s]$.

\section{Conclusions}
\label{sec:conclusions}
This paper provided an effective framework to construct exponent lattice basis in an inductive way. In many situations the framework can handle problems  larger than those that \textbf{FindRelations} can do. It is efficient especially when degrees of algebraic numbers to deal with can be intensively reduced, when the rank of algebraic numbers is small and when non-degenerate condition holds for many of the input algebraic numbers. The non-degenerate condition provides a relatively cheaper certificate for multiplicative independence. It usually holds for randomly picked algebraic numbers, thus it can be useful for proving their multiplicative independence. However, it often fails when numbers are algebraically dependent. Though the problem of computing exponent lattice basis for nonzero algebraic numbers is still difficult (the exponent lattice  basis for roots of a general rational polynomial of degree $5$ cannot be computed fast), our framework still casts light on how we may handle larger problems.

\section*{ACKNOWLEDGMENTS}
This work was supported partly by NSFC under grants 61732001 and 61532019. The authors thank Professor Shaoshi Chen for providing helpful references and Haokun Li for useful advice on programming.

\bibliographystyle{plain}

\def\cprime{$'$}

\end{document}